\documentclass{article}
\usepackage[utf8]{inputenc}
\usepackage{lmodern}
\usepackage{amsthm}
\usepackage{amsmath}
\usepackage{booktabs}
\usepackage{multirow}
\usepackage[T1]{fontenc}
\usepackage{diagbox}
\title{Longest (Sub-)Periodic Subsequence}
\author{Hideo Bannai  \and Tomohiro I \and Dominik K\"oppl}
\date{}

\usepackage{xcolor}
\definecolor{teigiIro}{HTML}{5700B5}
\newcommand*{\teigi}[1]{{\color{teigiIro}\emph{#1}}} 

\usepackage{amssymb}
\usepackage[pdf,theorems,natbib,cref]{nikis}
\newcommand*{\LCS}[2]{\ensuremath{{\LCSs}_{#1}({#2})}}
\newcommand*{\LCSs}{\ensuremath{\textrm{LCS}}}

\newcommand*{\Rat}{\ensuremath{\mathbb{Q}^+}}

\begin{document}

\maketitle

 \begin{abstract}
	 We present an algorithm computing the longest periodic subsequence of a string of length~$n$ in \Oh{n^7} time with \Oh{n^4} words of space.
	 We obtain improvements when restricting the exponents or extending the search allowing the reported subsequence to be subperiodic down to \Oh{n^3} time and \Oh{n^2} words of space.
 \end{abstract}
\section{Introduction}

A natural extension of the analysis of regularities such as squares or palindromes perceived as substrings of a given text is the study of the same type of regularities when considering subsequences.
In this line of research, given a text of length~$n$, \citet{kosowski04efficient} proposed an algorithm running in \Oh{n^2} time using \Oh{n} words of space to find the longest subsequence that is a square.
\citet{inoue18squaresubseq} generalized this setting to consider the longest such subsequence common of two texts~$T$ and $S$ of length~$n$, and gave an algorithm computing this sequence in \Oh{n^6} time using \Oh{n^4} space, 
also providing improvements in case that the number of matching characters pairs $T[i] = S[j]$ is rather small.
Recently, \citet{inoue20longest} provided similar improvements for the longest square subsequence of a single string.
Here, we consider the problem for a single text, but allow the subsequence to have different exponents.
In detail, we want to find the longest subsequence that is (sub-)periodic.

A non-exhaustive list of related problems are finding
the longest palindromic subsequence~\cite{chowdhury14computing,inenaga18hardness},
absent subsequences~\cite{kosche21absent},
longest increasing and decreasing subsequences~\cite{schensted61longest,duraj20subquadratic},
maximal common subsequences~\cite{sakai19maximal,conte19polynomial},
the longest run subsequence~\cite{schrinner20longest},
the longest Lyndon subsequence~\cite{bannai22computing}, 
longest rollercoasters~\cite{gawrychowski19rollercoasters,biedl19rollercoasters,fujita21longest},
and computing subsequence automaton~\cite{bille17subsequence,baeza-yates91searching}.
Our techniques rely on finding longest common subsequences, 
which is conceived as a well-studied problem (see~\cite{hirschberg75linear,hirschberg77algorithms,kiyomi21longest} and the references therein).

\section{Preliminaries}
Let $\N$ denote all natural numbers~$1,2,\ldots$, and $\Rat$ the set of all rational numbers greater than or equal to 1.
We distinguish integer intervals $1,\ldots,n = [1..n]$ and intervals of rational numbers $[1/2,3/4]$.

Let $\Sigma$ denote a totally ordered set of symbols called the \teigi{alphabet}.
An element of $\Sigma^*$ is called a \teigi{string}.
Given a string $S \in \Sigma^*$, we denote its length with $|S|$ and
its $i$-th symbol with $S[i]$ for $i \in [1..|S|]$.
Further, we write $S[i..j] = S[i]\cdots S[j]$.
A \teigi{factorization} $T = F_1 \cdots F_z$ is a partitioning of $T$ into substrings $F_1,\ldots,F_z$.
A \teigi{subsequence} of a string~$S$ with length~$\ell$ is a string $S[i_1] \cdots S[i_\ell]$ with $i_1 < \ldots < i_\ell$.

Given a string~$S$, we can write $S$ in the form $S = U^x U'$ with $U'$ being either empty or a proper prefix of $U$.
Then $|U|$ is called a \teigi{period} of $S$, and 
$x+|U'|/|U| \in \Rat$ is called its \teigi{exponent} with respect to the period~$|U|$.
For the largest possible such exponent~$x$,
$S$ is called \teigi{periodic} if $x \ge 2$, or \teigi{sub-periodic} if $x \in (1,2)$.
For instance, the unary string $T = \texttt{a}\cdots\texttt{a}$ has the minimum period~$1$ with exponent~$|T|$,
or more generally, period~$p \in [1..|T|]$ with exponent~$|T|/p$. 

Further, for two strings~$Y$ and $Z$, let \LCS{Y,Z}{y,z} denote the longest common subsequence of $Y[1..y]$ and $Z[1..z]$. 
We omit the strings in subscript if they are clear from the context. 
Also, we allow us to write $\LCSs$ for an arbitrary number of strings; 
the number of strings is reflected by the arguments given to $\LCSs$.
It is known that we can answer the longest common subsequence $\LCS{X_1,\ldots,X_k}{x_1,\ldots,x_k}$ of $k$ strings $X_1[1..x_1],\ldots,X_k[1..x_k]$
in constant time by building a table of size $\Oh{n^k}$ and filling its entries in \Oh{k n^k} time via dynamic programming.

\paragraph{Structure of the Paper}
In what follows, we first show an algorithm (\cref{secStarter}) computing the longest subsequence that is periodic or sub-periodic.
Subsequently, we refine this algorithm to omit the sub-periodic subsequences by allowing more time and space in \cref{secComputingPeriodic}.
\Cref{tableResults} gives an overview of our obtained results.
We observe that, if we are only interested in the longest subsequence having an exponent $> 1$,
we obtain a algorithm faster than those finding subsequences with more restricted exponents.

\begin{table}
	\caption{Space and time complexities for finding periodic subsequences of specific exponents.
		$\epsilon$ is a rational number with $0 < \epsilon < 1$.
		The exponent column means that a subsequence is considered only if it has at least one exponent within the domain given in the column.
	}
	\label{tableResults}
	\centerline{\begin{tabular}{llll}
			\toprule
			exponent & time & space & solution
			\\\midrule
			$2 \mathbb{N}$ & {\Oh{n^2}} & {\Oh{n}} & \cite{kosowski04efficient} \\
			$\N + \epsilon$ & {\Oh{n^3}} & {\Oh{n^2}} & \cref{thmSubPeriodic} \\
			$\Rat \cap ((2,3] \cup [4,\infty))$ & {\Oh{n^5}} & {\Oh{n^3}} & \cref{thmPeriodicThree}\\
$(3 + \epsilon) \mathbb{N}$ & {\Oh{n^7}} & {\Oh{n^4}} & \cref{thmPeriodicFour} \\
			\bottomrule
		\end{tabular}
	}\end{table}

A key observation is that a longest (sub-)periodic subsequence~$S$ 
is \teigi{maximal},
meaning that no occurrence $T[i_1]T[i_2]\cdots T[i_{|S|}]$ of $S$ in $T$ can be extended with a character 
in $T[1..i_1-1]$ or $T[i_{|S|}+1..]$ to form a longer subsequence without breaking the periodicity.

\section{Longest (Sub-)periodic Subsequence}\label{secStarter}
We start with the search for the longest subsequence that is periodic or subperiodic,
meaning that one of its
exponents is in $\Rat \setminus \N = (1,2)\cup (2,3) \cup (3,4)\cdots$.
The idea is to compute every possible factorization of $T = YZ$ into two factors, 
and try to (a) prolong each square $UU$ to $UcU$ for $Uc$ being a subsequence in $Y$,
or (b) extend $UU'$ to $UcU'$ with $U'$ being a proper prefix of $U$ found in $Z$ and $Uc$ a subsequence found in $Y$.
For a fixed partition $T = YZ$, we define

\begin{align}\label{eqSubSequence}
D_2[y,z] := \max 
\begin{cases}
	2 \cdot \LCS{Y,Z}{y-1,z}+ 1 &\text{~if~} \LCS{Y,Z}{y-1,z-1} > 0, \\
	D_2[y-1,z] + 1 &\text{~if~} D_2[y-1,z] > 0, \\
	D_2[y  ,z-1],\\
	0,
\end{cases}
\end{align}
for $y \in [1..|Y|]$ and $z \in [1..|Z|]$
to be the longest subsequence of $T$ having an exponent in $\Rat \setminus \N $ of the form $UU'$ 
with $U'$ being a non-empty common subsequence of $Y[1..y]$ and $Z[1..z]$ and $U$ a subsequence of $Y[1..y]$,
where $D_2[\cdot,0] := D_2[0,\cdot] := 0$.
Note that we can add $D_2[y-1,z-1] +1 \text{~if~} D_2[y-1,z-1] > 0$ as another selectable value for the maximum determining $D_2[y,z]$,
but this does not change the maximum, since the value of $D_2[y-1,z-1]+1$ is used as an option for the maximum determining $D_2[y,z-1]$, which is already an option for $D_2[y,z]$.
$D_2$ has the following property:

\begin{figure}[t]
	\centering{\ttfamily
\begin{tabular}{clllll}
			\toprule
			\backslashbox{$Z$}{$Y$}
			& \multicolumn{1}{c}{a} & \multicolumn{1}{c}{b} & \multicolumn{1}{c}{c} & \multicolumn{1}{c}{a} & \multicolumn{1}{c}{a}
			\\\midrule
			$Z[1] = $ a & $/$ & aba & abca & abcaa & abcaaa \\
			$Z[2] = $ c & $/$ & aba & abca & acaac & acaaac \\
			$Z[3] = $ a & $/$ & aba & abca & acaac & acaaaca \\
			\bottomrule
		\end{tabular}
	}\caption{$D_2$ of \cref{secStarter} for the text $T = \texttt{abcaaaca}$ with the factorization $T = YZ$ with
	$Y= \texttt{abcaa}$ and $Z = \texttt{aca}$.}
	\label{fig}
\end{figure}

\begin{lemma}\label{lemSubPeriodic}
	$D_2[y,z]$ is the length of the longest subsequence of $T$ of the form $UU'$ with $U'$ being a proper prefix of $U$ and a common subsequence of $Y[1..y]$ and $Z[1..z]$, and $U$ being a subsequence of $Y[1..y]$.
\end{lemma}
\begin{proof}
	Let $U U'$ be the longest subsequence having an exponent in $\Rat \setminus \N $ in $T$
	with $Y[y_1] \cdots Y[y_{|U|}] Z[z_1] \cdots Z[z_{|U'|}] = U U'$.
	We partition $T = YZ$ such that $U$ and $U'$ are subsequences in $Y$ and $Z$, respectively.
	Since $U'$ is a proper prefix of $U$, $\LCS{Y,Z}{y_{|U'|},z_{|U'|}} \ge |U'|$.
	If $\LCS{Y,Z}{y_{|U'|},z_{|U'|}} > |U'|$, then there is another common subsequence $V'$ of $Y$ and $Z$,
	and $V' U[|U'|..]$ is a longer subsequence of $Y$ than $U$, so
	$V' U[|U'|..] U'$ is a longer subsequence having an exponent in $\Rat \setminus \N$ than $U U'$, a contradiction.
	Therefore, $\LCS{Y,Z}{y_{|U'|},z_{|U'|}} = |U'|$.
	Obviously, $Y[y_{|U'|}..] = U[y_{|U'|}..]$, 
	since otherwise we could extend $U$ further.

	Finally, we show that the sequence $Y[y_1] \cdots Y[y_{|U|}] Z[z_1] \cdots Z[z_{|U'|}]$ is considered in the construction of $D_2$.
	Since $y_{|U'|}+1 \le y_{|U'|+1} \le |Y|$ (otherwise $U'$ cannot be a proper prefix of $U$), 
	$D_2[y_{|U'|}+1,z_{|U'|}] \ge 2\cdot\LCS{Y,Z}{y_{|U'|},z_{|U'|}} + 1$,
	and equality results from the fact that we otherwise have found a longer subsequence 
	$V V'$ in $Y[1..y_{|U'|}+1]Z[1..z_{|U'|}]$ that we can extend to a subsequence of $YZ$ longer than $U U'$ by applying the second option in \cref{eqSubSequence}.
	The third option of \cref{eqSubSequence} fills up $D_2$ such that we obtain the length of $U U'$ in $D_2[|Y|,|Z|]$.
\end{proof}

\begin{theorem}\label{thmSubPeriodic}
	We can find the longest subsequence having an exponent in $\Rat \setminus \N$ in \Oh{n^3} time using \Oh{n^2} space.
\end{theorem}
\begin{proof}
	For each of the $n$ different partitions $T = YZ$,
	we precompute a table answering $\LCS{Y,Z}{y,z}$ in constant time.
	This table needs \Oh{n^2} space and can be constructed in \Oh{n^2} time.
	Next, we create a table $D_2$ of size \Oh{n^2}, and fill each of its cells by \cref{eqSubSequence} in constant time thanks to the precomputation step.
	In total, we need \Oh{n^2} time per partition $T = YZ$, and therefore \Oh{n^3} time for the entire computation.
\end{proof}

In the following we want to omit subsequences having exponents only in $[1,2)$.

\section{Longest Periodic Subsequence}\label{secComputingPeriodic}
We now extend our ideas of the previous section to omit the sub-periodic subsequences,
such that our algorithm always reports a subsequence with an exponent of at least $2$.
Our main idea is to generalize the factorization of $T$ from $2$ to $k$ factors.
	Computing the $\LCSs$ of these $k$ factors,
		we obtain all longest periodic subsequences having an exponent of length $k\ell$ with $\ell \in \mathbb{N}$.
	For $k=2$, we can find all square subsequences, i.e., the longest common subsequence with an exponent of $2x$ for $x \in \mathbb{N}$, similar to \cite{kosowski04efficient}.
		Like in \cref{secStarter}, we can support exponent values $\ell \in \Rat$ by stopping matching characters in the last factor.
	In general, the number of factors $k = 3$ is a good value if the exponent is not in $(3,4)$.
		With an exponent $x \in (3,4)$, 
		each factor starts capturing at the root of the repetition, i.e.,
		if we match a subsequence $S = U^{\gauss{x}} U'$, then we start capturing $U$ in all factors simultaneously.
		However, the last factor has to capture $|U U'|$ characters if $x \in (3,4)$.
		Hence, we need to split this last factor such that we have four factors, i.e., we need $k=4$ for $x \in (3,4)$.

	However, $k = 3$ suffices for $x \in (2,3] \cup [4,\infty)$.
		To see that, we let the first $k-1$ factor capture the subsequence $U^{\gauss{x/k-1}}$.
		The last factor captures $U^y$ with $y = x - \gauss{x/(k-1)} (k-1)$, which works if $y \le \gauss{x/(k-1)}$,
		i.e., $x \le 3 \gauss{x/2}$, which holds for $x \ge 4$.
		For $x \in (2,3]$, each of the first factors captures $U$, while the last factor captures $U^y$ with $y \le 1$.

\subsection{Three Factors}
There are ${n \choose 2} = \Oh{n^2}$ possible factorizations of the form $T = XYZ$.
Let us fix one such factorization $T = XYZ$.
For this factorization, we define the three-dimensional table $D_3[1..|X|,1..|Y|,1..|Z|]$ by
\[
D_3[x,y,z] := \max 
\begin{cases}
	3 \cdot \LCS{X,Y,Z}{x-1,y-1,z} + \delta_{xy}\\
D_3[x-1,y-1,z]+\delta_{xy}\\
D_3[x-1,y  ,z]\\
D_3[x  ,y-1,z]\\
D_3[x  ,y  ,z-1]\\
\end{cases}\]
where $\delta_{xy} := 2$ if $X[x]=Y[y]$ and $0$ otherwise.
Compared to $D_2$, the number options for a cell value have increased.
The first option takes the longest common subsequence $U$ of $X[1..x-1]$, $Y[1..y-1]$ and $Z[1..z]$, which induces the subsequence $U^3$ of $T$,
and tries to prolong $U^3$ to a subsequence $(Uc)^2 U'$ of $T$ with $U'$ being a prefix of $U$ and $c = X[x] = Y[y]$ 
(or gives $U^3$ if $X[x] \not= Y[y]$).
The second option performs an extension of $U^2 U'$ to $(Uc)^2 U'$ with $c = X[x] = Y[y]$.
The last options are similar to the standard $\LCSs$ computation by copying a previously computed result on a mismatch of $X[x]$ and $Y[y]$.

\begin{lemma}\label{lemPeriodicThree}
	$D_3[x,y,z]$ is the longest subsequence $U^2 U'$ of $T$ with $U'$ (possibly empty) prefix $U$ 
	such that $U$ is a common subsequence of $X[1..x]$, and $Y[1..y]$,
	and $U'$ is a subsequence of $Z[1..z]$.
\end{lemma}
\begin{proof}
	The longest periodic subsequence with an exponent in $3\mathbb{N}$ is given by $\LCS{X,Y,Z}{|X|,|Y|,|Z|}$.
	In what follows, we show that we can also find the longest periodic subsequence with an exponent not divisible by three,
	i.e., the exponent is in $I := \Rat \cap ((2,3] \cup [4,\infty)) \setminus 3\mathbb{N}$.
	This longest periodic subsequence has the form $U^2 U' = X[x_1]\cdots X[x_{|U|}]Y[y_1]\cdots Y[y_{|U|}]Z[1]\cdots Z[z_{|U'|}]$ with $U'$ being a (not necessarily proper) prefix of $U$ (see the aforementioned discussion of why this subsequence has to admit this form).
	Then $U'$ is a common subsequence of $X$, $Y$, and $Z$ with $\LCS{X,Y,Z}{x_{|U'|},y_{|U'|},z_{|U'|}} \ge |U'|$.
	Similar to the proof of \cref{lemSubPeriodic}, we can argue that 
	$\LCS{X,Y,Z}{x_{|U'|},y_{|U'|},z_{|U'|}} = |U'|$, otherwise $T$ has a longer periodic subsequence with an exponent in $I$.
	Hence, the longest periodic subsequence with an exponent in $I$ has the form $V^2 U'$ with $U'$ being a prefix of $V$.
	But then $V[|U'|..]$ is the longest common subsequence of $X[x_{|U'|}..]$ and $Y[y_{|U'|}..]$,
	computed in $D_3$.
\end{proof}

\begin{theorem}\label{thmPeriodicThree}
	We can find the longest periodic subsequence with an exponent in $\Rat \cap ((2,3] \cup [4,\infty))$ in \Oh{n^5} time using \Oh{n^3} space.
\end{theorem}
\begin{proof}
	For the claimed time and space complexities, 
	let us fix a factorization $T = XYZ$.
	We first pre-process $\LCSs$ with a three-dimensional table taking $\Oh{n^3}$ space and \Oh{n^3} time such that we can answer an $\LCSs$ query in $\Oh{1}$ time.
 The table $D_3$ also takes \Oh{n^3} space, and each cell can be filled in constant time thanks to the pre-processing step.
 Finally, we compute the maximum value of $D_3$ for each factorization $T = XYZ$, which are \Oh{n^2} many.
 Hence, we fill $D_3$ \Oh{n^2} times, which gives us the total time of \Oh{n^5}.
\end{proof}

\subsection{Four Factors}

Finally, we consider a factorization of size four to capture exponents in $(3,4)$.
We have ${n \choose 3} = \Oh{n^3}$ possibilities to factorize $T$ into four factors~$W,X,Y,Z$.
Let us fix a factorization $T = WXYZ$.
We fill the 4-dimensional table $D_4[1..|W|,1..|X|,1..|Y|,1..|Z|]$ as follows:
\begin{align*}
D_4[w,x,y,z] := \max 
\begin{cases}
	4 \cdot \LCS{W,X,Y,Z}{w-1,x-1,y-1,z} + \delta_{wxy} & \text{if~} \delta_{xyz} > 0 \\
D_4[w-1,x-1,y-1,z] + \delta_{wxy} \\
D_4[w-1,x  ,y  ,z]\\
D_4[w  ,x-1,y  ,z]\\
D_4[w  ,x  ,y-1,z]\\
D_4[w  ,x  ,y  ,z-1]
\end{cases}
\end{align*}
where $\delta_{wxy} := 3$ if $W[w]=X[x]=Y[y]$ and $0$ otherwise.

\begin{theorem}\label{thmPeriodicFour}
	We can compute the longest periodic subsequence with an exponent $\in \bigcup_{x \in \N} (3x,4x) \cap \Rat$ in \Oh{n^7} time using \Oh{n^4} space.
\end{theorem}
\begin{proof}
	We can show analogously to \cref{lemPeriodicThree} that
	$D_4[w,x,y,z]$ stores the length of the longest string $U^3 U'$ for which $U'$ is both a proper prefix of $U$ and a subsequence of $Z[1..z]$, while $U$ is a common subsequence of the three strings $W[1..w]$, $X[1..x]$, and $Y[1..y]$.
	The rest can be proven analogously to \cref{thmPeriodicThree} by adding an additional dimension.
\end{proof}

\section{Open Problems}
We are unaware of polynomial-time algorithms computing several other types of regularities when considering subsequences.
For instance, we are not aware of an algorithm computing the longest \emph{sub-periodic} subsequence.
For that, we would need an efficient algorithm computing the longest (common) subsequence \emph{without a border}. 
Then we could use the algorithm of \cref{secStarter} and compute this subsequence without a border instead of the (plain) $\LCSs$.
Other problems are finding the 
longest (common) subsequence that is \emph{primitive} (exponent in $(1,\infty) \setminus \mathbb{N}$), or
	the
	longest (common) subsequence that is \emph{non-primitive} (exponent in $\mathbb{N}\setminus\{1\}$).

\bibliographystyle{abbrvnat}

\end{document}